\documentclass[fullpage]{article}

\usepackage{theorem,latexsym,graphicx}
\usepackage{amsmath,amssymb,enumerate}
\usepackage{xspace}
\usepackage{fullpage}
\usepackage{bm}
\usepackage{ifpdf}
\usepackage{shadow,shadethm,color}
\usepackage{algorithm}

\usepackage{algorithm}
\usepackage[noend]{algpseudocode}

\usepackage{subfigure}
\usepackage{verbatim}
\usepackage{paralist}

\allowdisplaybreaks

\definecolor{Darkblue}{rgb}{0,0,0.4}
\definecolor{Brown}{cmyk}{0,0.81,1.,0.60}
\definecolor{Purple}{cmyk}{0.45,0.86,0,0}
\newcommand{\mydriver}{hypertex}
\ifpdf
 \usepackage{pdfpages}
 \renewcommand{\mydriver}{pdftex}
\fi
\usepackage[breaklinks,\mydriver]{hyperref}
\hypersetup{colorlinks=true,
            citebordercolor={.6 .6 .6},linkbordercolor={.6 .6 .6},%
citecolor=Darkblue,urlcolor=black,linkcolor=Darkblue,pagecolor=black}
\newcommand{\lref}[2][]{\hyperref[#2]{#1~\ref*{#2}}}

\newtheorem{theorem}{Theorem}[section]

\newtheorem{lemma}[theorem]{Lemma}
\newshadetheorem{lemmashaded}[theorem]{Lemma}

\numberwithin{algorithm}{section}

\newenvironment{proof}{{\bf Proof:  }}{\hfill\rule{2mm}{2mm}}

\newcommand{\junk}[1]{}
\newcommand{\ignore}[1]{}

\newcommand{\N}[0]{{\ensuremath{\mathbb{N}}}}

\def\floor#1{\lfloor #1 \rfloor}
\def\ceil#1{\lceil #1 \rceil}

\newcommand{\E}{{\mathbf{E}}}

\newcommand{\Ex}{\E}
\newcommand{\one}{\mathbf{1}}
\newcommand{\rest}{\ensuremath{rest}}
\newcommand{\pder}[2]{\frac{\partial#1}{\partial#2}}
\newcommand{\pdertwo}[2]{\frac{\partial^2#1}{\partial{#2}^2}}

\newcounter{note}[section]

\newcommand{\agnote}[1]{} 
\newcommand{\ktnote}[1]{}
\newcommand{\dwnote}[1]{}
\newcommand{\uwnote}[1]{}

\newcommand{\ud}{\texttt{d}}

\begin{document}


\title{Consistent Weighted Sampling Made Fast, Small, and Easy}

\author{Bernhard Haeupler\\Carnegie Mellon University\\Haeupler@cs.cmu.edu \and Mark Manasse\\msanasse@gmail.com \and Kunal Talwar\footnote{This research was performed while all the authors were at Microsoft Research Silicon Valley.}
\\k@kunaltalwar.org}
\date{}
\maketitle
\begin{abstract}
Document sketching using Jaccard similarity has been a workable effective technique in reducing near-duplicates in Web page and image search results, and has also proven useful in file system synchronization, compression and learning applications~\cite{BroderGMZ97,Broder97,BroderCFM98}. 

Min-wise sampling can be used to derive an unbiased estimator for Jaccard similarity and taking a few hundred independent consistent samples leads to compact sketches which provide good estimates of pairwise-similarity. 
Early sketching papers handled weighted similarity, for integer weights, by transforming an element of weight $w$ into $w$ elements of unit weight, each requiring their own hash function evaluation in the consistent sampling. Subsequent work~\cite{GollapudiP,ManasseMT,Ioffe10} removed the integer weight restriction, and showed how to produce samples using a constant number of hash evaluations for any element, independent of its weight. Another drastic speedup for sketch computations was given by Li, Owen and Zhang~\cite{LiOZ12} who showed how to compute such (near-)independent samples in one shot, requiring only a constant number of hash function evaluations per element. Unfortunately this latter improvement works only for the unweighted case. 

In this paper we give a simple, fast and accurate procedure which reduces weighted sets to unweighted sets with small impact on the  Jaccard similarity. This leads to compact sketches consisting of many (near-)independent weighted samples which can be computed with just a small constant number of hash function evaluations per weighted element. The size of the produced unweighted set is furthermore a tunable parameter which enables us to run the unweighted scheme from~\cite{LiOZ12} in the regime where it is most efficient. Even when the sets involved are unweighted, our approach gives a simple solution to the densification problem that~\cite{ShrivastavaL14a,ShrivastavaL14b} attempt to address. 

Unlike previously known schemes, ours does not result in an unbiased estimator. However, we prove that the bias introduced by our reduction is negligible and that the standard deviation is comparable to the unweighted case. We also empirically evaluate our scheme and show that it gives significant gains in computational efficiency, without any measurable loss in accuracy.
\end{abstract}

\section{Introduction}

Web experiments have repeatedly shown that most breadth-first collections of pages contain many unique pages, but also contain large clusters of near-duplicate pages.  Typical studies have found that duplicate and near-duplicate pages account for between a third and a half of a corpus.

Min-wise sampling~\cite{Broder97,BroderGMZ97,BroderCFM98} has been widely used in Web and image search since the mid-nineties to produce consistent compact sketches which provide good estimates of pairwise similarity of corpus items, computing the sketch with reference only to a single item.  Min-wise hashing computes a sketch for estimating the Jaccard (scaled L1) similarity effectively; SimHash~\cite{Charikar02}, and related techniques compute sketches for estimating the angular separation (in L2)  of arbitrary vectors. Both of these have been widely used in deployed commercial search engines to allow the search result pages to suppress reporting the near-duplicate pages which would otherwise often dominate the search results.

SimHash, by its nature, allows vector coordinates to be arbitrary real numbers, and weights the random projections accordingly. While min-wise sampling is designed for unweighted sets it can also be used for non-negative integer weights. For this one simply replaces any element $e$ of weight $w$ by $w$ elements $e_1,\ldots,e_w$ of unit weight. This reduction however leads to the running time of computing one consistent sample to be proportional to the sum of all weights. More recent papers~\cite{GollapudiP,ManasseMT,Ioffe10} removed the integer weight restriction, and showed how to produce samples using a constant number of hash evaluations for any element, independent of its weight. 

All these sampling techniques as described typically result in a Boolean random variable whose expectation is related to the similarity. One typically applies the same procedure repeatedly with independent randomness to produce a sketch consisting of hundreds to samples, to get  independent Boolean estimates that can be averaged. In this work, we will be concerned with designing a faster estimation scheme for Jaccard similarity.

A beautiful idea of Li and K\"onig~\cite{LiK11,li2010b}, known as $b$-bit min-wise hashing, helps reduce the size of a sketch by storing only a $b$-bit hash per sample. While this results in some ``accidental'' hash collisions, this can be remedied by taking into account the effect of these collisions and taking a larger number of samples. This gives more compact sketches but may require a longer time for sketch generation.

For the unweighted case Li, Owen and Zhang~\cite{LiOZ12} showed how to drastically speed up the computation of sketches by computing 200, say, (near-)independent samples in one shot using only a constant number of hash function evaluations per element, instead of computing each sample one-by-one. Unfortunately however, this ``one permutation''-technique does not easily extend to weighted sampling.

In this paper, we bring this level of performance to weighted sampling, producing an algorithm which produces a sketch in time proportional to the number of positive-weight elements in the item. We do this by picking two or more scales and converting the weighted set into an unweighted one by randomized rounding. This sampling step introduces a negligible error and bias and leads to a small unweighted set on which any unweighted sketching technique can be applied. The size of this set is a tunable parameter which can be beneficial for the subsequently used unweighted sketching step. We apply this new algorithm, and the older ones, in a variety of settings to compare the variance in accuracy of approximation. These improvements come at a marginal cost. Our algorithm takes as input an interestingness threshold $\alpha$, say $\alpha = \frac{1}{2}$, such that similarities smaller than $\alpha$ are considered uninteresting. Given two weighted sets, it either returns an accurate estimate of the similarity, or correctly declares that the similarity is below $\alpha$. Since in applications, one is not usually interested in estimating similarity when it is small, we believe that this is an acceptable tradeoff. 

\section{Background}
Given two finite sets $S$ and $T$ from a universe $U$, the {\em Jaccard similarity} of $S$ and $T$ is defined as:
\begin{align*}
jacc(S,T) = \frac{|S \cap T|}{|S \cup T|}.
\end{align*}

A weighted set associates a positive real weight to each element in it. Thus a weighted set is defined by a map $w : U \rightarrow \Re_+$, with the weight of the elements outside the set defined as 0. We denote the support of $w$ by $supp(w) = \{a \in U: w(a) >0\}$. An unweighted set is then the special case where the weight is equal to one on all of its support.

This definition of Jaccard similarity can be extended to weighted sets in a natural way.
Given two mappings $w_S$ and $w_T$ with supports $S$ and $T$ respectively, 
their {\em weighted Jaccard similarity} $jacc(W,V)$ is defined as
\begin{align*}
\frac{\sum_{a \in S \cap T} \min(w_S(a), w_T(a))}{\sum_{a \in S \cup T} \max(w_S(a), w_T(a))} = \frac{| \min(w_S,w_T) |_1}{| \max(w_S,w_T) |_1}.
\end{align*}

Given a weighted set $w$, we denote by $n(w)$ the support size $|supp(w)|$, and by $W(w)$ the total weight $|w|_1$. When the weighted set if clear from context, we will simply use $n$ and $W$ to denote these quantities. We will be concerning ourselves with fast algorithms for creating a short sketch of a weighted set from which we can quickly estimate the Jaccard similarity between two sets.

Empirical observations place Jaccard similarity below $\approx$ 0.7 as probably not near-duplicates, and above $\approx$ 0.95 as likely near- or exact-duplicates. We will be concerned with fast and accurate sketching techniques for Jaccard similarity. 

Sampling techniques which pick without replacement can closely approximate Jaccard: the oldest such, working in a limited setting where weights are integers (or integer multiples of some fixed base constant) replaces item $a$ with weight $w(a)$ by new items $(a,1), \ldots, (a,w(a))$, all of weight 1.  By taking a hash function $h$, and computing it multiple times for each input item, we can map each pair of item, sample number to a positive real number.  To select the $j^{th}$ sample, consider ${h(x, j)}$ for all items $x$, and choose the pair producing the numerically smallest hash value.  This leads to the generation of $k$ samples in time $O(Wk)$.

Manasse, McSherry and Talwar~\cite{ManasseMT} extended this classic scheme to the case of arbitrary non-negative weights, using the ``active index'' idea of Gollapudi and Panigrahy~\cite{GollapudiP}. Their scheme had the additional advantage that only expected constant time per element is required, independent of its weight, which leads to a total expected run-time of $O(nk)$ for generating $k$ samples, independent of $W$. Ioffe~\cite{Ioffe10} improved this to worst-case constant time, by carefully analyzing the resulting distributions, reducing the per input per sample cost to choosing five random uniform values in the range between zero and one. This gives a worst case run-time of $O(nk)$.

In a deployed implementation, the principal costs for this per element are: (a) Seeding the pseudo-random generator with the input, (b) computing roughly 200 sets of 5 random values in $[0,1]$, and manipulating them to compute a 200 hash values, and (c)   comparing these hash values to a vector of the 200 smallest values discovered to date, and replacing that value if smaller.

Li, Owen and Zhang~\cite{LiOZ12}, using a technique first explored by Flajolet and Martin~\cite{FlajoletM85} compute a single hash value for each element, as well as a random sample number; the hash value then contends for smallest only among those values with equal sample number. This approach does not seem to extend to weighted sampling, because an item with very large weight may need to contend for multiple samples in order to sufficiently influence the predicted similarity. Another deficiency of this approach is that when sets are small, the accuracy of the estimator suffers. Recent works by Shrivastava and Li~\cite{ShrivastavaL14a,ShrivastavaL14b} attempt to address the later concern. 

More recently Li and K\"onig~\cite{LiK11,li2010b} found it effective to instead store up to 256 smallest values, but store only 1 or 2 bits derived randomly from the smallest values.  Accidental collisions will happen a quarter or a half the time, but we can get equivalent power for estimating the Jaccard value by computing enough extra samples to account for the matches that occur due to insufficient length of the recorded value.  For 2 bit samples, with 136 samples, we expect 34 to match randomly, leaving us with 102 samples that match with probability equal to the Jaccard value; na\"ive Chernoff bounds allow us to conclude that estimates of the true value will be accurate to within 0.1. For 1 bit samples, drawing 200 results in 100 accurate samples; the storage space required without these modifications is on the order of 800 bytes. The 2 bit variant takes 34 bytes to hold (with 800 bytes needed in memory prior to the 2 bit reduction, since the minimization step needs to done accurately), while the 1 bit variant needs 25 bytes; a further variant instead computes 400 results, reduces each to 1 bit, and then computes the exclusive-or of pairs of bits to reduce back to 200 bits, each of which will match randomly half the time; this gives us an estimator for $p^2$, rather than $p$, which we can then take the square root of.

In our new algorithm, we seek to gain for the weighted case the space efficiency of~\cite{LiK11}, while producing a sample set as efficiently as~\cite{LiOZ12} do for the unweighted case.

\subsection*{Other Related Work}
The Jaccard similarity is an $\ell_1$ version of similarity between two objects. The cosine similarity is an $\ell_2$ notion that has also been used, and SimHash gives a simple sketching scheme for it. Henzinger~\cite{Henzinger06} performed an in-depth comparison of the then state-of-the-art algorithms for MinHash and SimHash. More recently, Shrivastava and Li~\cite{ShrivastavaL14c} compare these two approaches in some other applications. 
While as stated, SimHash requires time $O(nk)$ to generate $k$ samples, recent work on fast Johnson-Lindenstrauss lemma~\cite{AilonC09,AilonL09,DasguptaKS10} may be viewed as giving faster algorithms for SimHash.

\section{Algorithm rationale and design}

A sketching based similarity estimation scheme consists of two subroutines:
\begin{itemize}
\item A sketching algorithm $Sketch$ that takes as input a single (possibly) weighted set (and usually, a common random seed) and returns a sketch, and 
\item an estimating algorithm $Estimate$ that takes as input two sketches generated by the sketching and returns an estimate of the Jaccard similarity.
\end{itemize}

The property one wants from this pair of algorithms is that for any pair of weighted sets $w_1$ and $w_2$, the estimate $Estimate(Sketch(w_1,r),Sketch(w_2,r))$ is ``close'' to the true Jaccard similarity, with high probability over the randomness $r$. Moreover, we want the run-time and the output size of the Sketch algorithm to be small.

We will present our algorithm in two parts. We first describe a reduction that given a weighted set $w$ and a random seed $r$ outputs an {\em unweighted} set $ReduceToUnwtd(w,r)$ such that:
\begin{description}
\item{(a)} the expected size of the unweighted set is $|w|_1$, and
\item{(b)} given any two weighted sets $w_1$ and $w_2$, the resulting unweighted sets $S_1 = ReduceToUnwtd(w_1,r)$ and $S_2 = ReduceToUnwtd(w_2,r)$ satisfy the property that $jacc(S_1,S_2)$ is approximately $jacc(w_1,w_2)$ with high probability, as long as $|w_1|_1$ and $|w_2|_1$ are large enough. 
\end{description}
We will formalize these statements in the next section. We will then describe how such a reduction can be used along with an unweighted similarity estimation scheme to generate fast and small sketches for weighted sets. In this second part, we will assume that we are given a threshold $\alpha$ such that similarities smaller than $\alpha$ need not be estimated accurately.

\subsection{Weighted to Unweighted Reduction}

In this section, we describe a simple randomized reduction that transforms any weighted set to an unweighted one, such that the Jaccard similarity between sets is approximately preserved. Given a weighted set $w$ from a universe $U$, the reduction produces an unweighted set $S \subseteq U \times \N$. We will assume access to a hash function $h$ that takes a random seed $r$ and a pair $(a,i) \in U \times \N$, and returns a real number in $[0,1)$. We assume for the proofs, that for a random $r$, the hash value $h(r,a,i)$ is uniform in $[0,1)$ and independent of $h(r,a',i')$ for any $(a',i')$ different from $(a,i)$. In practice, we will only need a small expected number of bits of this $h$, and using a pseudorandom generator will suffice.

The reduction is a simple randomized rounding scheme. Consider an element $a$ with weight $w(a)$. We write $w(a)= j_a + f_a$, where $j_a= \lfloor w(a) \rfloor$ is the integer part of $w(a)$ and $f_a \in [0,1)$ is the fractional part. We add to $S$ an element $(a,i)$ for $i=1,\ldots j_a$. Additionally, we add $(a,j(a)+1)$ with probability exactly $f_a$; we do this by computing a hash $h(r,a,j_a)$ and adding $(a,j_a+1)$ to $S$ if and only if $h(r,a,j_a) < f_a$. Using the same hash function seeded with the element $a$ ensures consistency in our rounding, which is crucial for (approximately) preserving Jaccard similarity. The resulting {\em ReduceToUnwtd} algorithm looks as follows:

\begin{algorithm}[htb!]
\caption{ReduceToUnwtd($w$,$r$)}
\begin{algorithmic}[1]

\State $S = \emptyset$
\For {$a$ with $w(a) > 0$}
		\State $j_a = \floor{w(a)}$
		\If {$h(r,a,j_a) < w(a) - j_a$}
			\State $j_a = j_a + 1$
		\EndIf
		\State $S = S \cup \bigcup_{p=1}^{j_a} \{(a,p)\}$
	\EndFor
\State Output $S$

\end{algorithmic}
\label{alg:reduce}
\end{algorithm}

This {\em ReduceToUnwtd} algorithm furthermore provides the following guarantees, whose proof we defer to the next section.

\begin{theorem}
\label{thm:reduction}
Let $w_1$ and $w_2$ be weighted sets and let $S_1 = ReduceToUnwtd(w_1,r)$ and $S_2 = ReduceToUnwtd(w_2,r)$ for a random seed $r$. Then for each $i=1,2$ and any $\delta > 0$,
\begin{enumerate}
\item{(Size Expectation)} $\Ex[|S_i|] = |w_i|_1$.
\item{(Size Tail)} $\Pr[\big| |S_i| - |w_i|_1 \big|  \geq 3\sqrt{|w_i|_1 \ln \frac{2}{\delta}}] \leq \delta$.
\end{enumerate}
Further let $W = \max(|w_1|_1, |w_2|_1)$. Then 
\begin{enumerate}
\setcounter{enumi}{2}
\item{(Bias)} $\big|\Ex[jacc(S_1,S_2)] - jacc(w_1,w_2)\big| \leq \frac{1}{W-1}$.
\item{(Tail)} $\Pr[|jacc(S_1,S_2) - jacc(w_1,w_2)| \geq \sqrt{\frac{27\ln \frac{4}{\delta}}{W}}] \leq \delta$.
\end{enumerate}
\end{theorem}

A variant of this algorithm will be useful when we want to even further improve on the number of hash computations that need to be performed. In particular, we propose the algorithm {\em ReduceToUnwtdDep} which uses $h(r,a)$ instead $h(r,a,j_a)$ to determine whether $(a,j_a+1)$ is added. Except for this small change in Step 4 the algorithm is identical to {\em ReduceToUnwtd}. 

\begin{algorithm}[htb!]
\caption{ReduceToUnwtdDep($w$,$r$)}
\begin{algorithmic}[1]

\State $S = \emptyset$
\For {$a$ with $w(a) > 0$}
		\State $j_a = \floor{w(a)}$
		\If {$h(r,a) < w(a) - j_a$}
			\State $j_a = j_a + 1$
		\EndIf
		\State $S = S \cup \bigcup_{p=1}^{j_a} \{(a,p)\}$
	\EndFor
\State Output $S$

\end{algorithmic}
\label{alg:reducedep}
\end{algorithm}

This slight change in {\em ReduceToUnwtdDep} compared to {\em ReduceToUnwtd} introduces some dependencies in the rounding. For example, if $w_a=1.5$ and $w'_a=2.5$, then the outputs of {\em ReduceToUnwtdDep} on these two weight functions $S$ and  $S'$ will have the events $(a,2) \in S$ perfectly correlated with the event $(a,3) \in S'$, whereas in the original {\em ReduceToUnwtd} algorithm these events are independent. Nevertheless, the following theorem, which is also proved in the next section, shows that the Jaccard similarity of $S$ and $S'$ is still close to that between $w$ and $w'$.

\begin{theorem}
\label{thm:reductiondep}
Let $w_1$ and $w_2$ be weighted sets and let $S_1 = ReduceToUnwtd(w_1,r)$ and $S_2 = ReduceToUnwtd(w_2,r)$ for a random seed $r$. Then for each $i=1,2$ and any $\delta > 0$,
\begin{enumerate}
\item{(Size Expectation)} $\Ex[|S_i|] = |w_i|_1$.
\item{(Size Tail)} $\Pr[\big| |S_i| - |w_i|_1 \big|  \geq 3\sqrt{|w_i|_1 \ln \frac{2}{\delta}}] \leq \delta$.
\end{enumerate}
Further let $W = \max(|w_1|_1, |w_2|_1)$. Then 
\begin{enumerate}
\setcounter{enumi}{2}
\item{(Bias)} $\big|\Ex[jacc(S_1,S_2)] - jacc(w_1,w_2)\big| \leq \frac{1}{W-1}$.
\item{(Tail)} $\Pr[|jacc(S_1,S_2) - jacc(w_1,w_2)| \geq \sqrt{\frac{27\ln \frac{4}{\delta}}{W}}] \leq \delta$.
\end{enumerate}
\end{theorem}

\subsection{Similarity Estimation Scheme}

We start by observing that the weighted Jaccard similarity is scale invariant: for any real number $\gamma>0$ it holds that $jacc(\gamma w_1, \gamma w_2) = jacc(w_1,w_2)$. For our scheme however different choices of $\gamma$ lead to a different outcome. In particular, our reduction gets more accurate as the total $\ell_1$ weight $w(U)$ increases. On the other hand, the expected size of the unweighted set resulting from the reduction is $w(U)$ so that any unweighted similarity estimation sketch we use gets more inefficient as $w(U)$ increases. We therefore would like to pick $\gamma$ to be as small as possible while keeping the accuracy loss in the reduction small.

A more pressing matter though is the following: if we know the sets $w_1$ and $w_2$, we can carefully pick $\gamma$, but the whole point of the sketch is that it summarizes $w_1$ without knowing which $w_2$ we would want to compare it with. Thus, we will need to decide on one or more scaling factors $\gamma$ for a set $w$ without knowing which other weighted sets we will compare it with.

To describe our scheme, we will introduce a few more parameters. The input parameter $\alpha<1$ is the interestingness threshold, and our scheme may report "$< \alpha$" instead of outputting an estimate if the Jaccard similarity is smaller than $\alpha$. The parameter $k$ will correspond to the final number of comparable samples (or $b$-bit hash values) for a pair of weighted sets of interest. This determines the accuracy of the scheme, which is of the order $\frac{1}{\sqrt{k}}$, the standard deviation of $k$ independent random measurements. Thus to get accuracy about $0.05$ in the estimate of the Jaccard similarity, we will use $k$ to be about $400$. An additional parameter $L$ is the redundancy we want to use when using  an unweighted similarity estimation scheme. Roughly speaking, for generating $k$ samples, we will assume that the unweighted similarity scheme we use works well on sets of size at least $Lk$. For a scheme such as the one we use~\cite{LiK11}, $L$ being a small constant such as $5$ suffices. Finally, a parameter $\beta$ will determine the scaling factors we use, and $t$ will denote the number of scales we use for each weighted set. For simplicity, we first describe the algorithm with $\beta=\alpha$.

We propose to pick for a weighted set $w$, a small number $t$ of ``scales'', where a scale is simply an integer power of $\beta$. We pick the scales in such a way that the scaled $\ell_1$ weight $\beta^i \cdot w(U)$ is in the range $[\frac{Lk}{t-1},\beta^{-1}\frac{Lk}{t-1})$ for the first scale, and in adjacent larger geometric intervals for the remaining $t-1$ scaling factors. This can be achieved by taking the first scaling factor to be $\beta^{-s}$ for $s$ given by $s=\lceil \log_{1/\beta} \frac{Lk}{(t-1)w(U)}\rceil$, and the subsequent scaling factors being given by $s+1,s+2,\ldots,s+t-1$. For each of these scales $s'$, we define the scaled weighted sets $\beta^{-s'}\cdot w$, and apply the reduction to it to derive an unweighted set of size at least $Lk/(t-1)$, and apply an unweighted sketching scheme to derive $\frac{k}{t-1}$ samples.

Observe that if $w_1(U)/w_2(U) \not\in [\alpha,1/\alpha]$, then the Jaccard similarity $jacc(w_1,w_2) < \alpha$. Thus for such sets, we can safely report ``similarity $<$ $\alpha$''. If on the other hand, the ratio $w_1(U)/w_2(U) \in [\alpha,1/\alpha]$, then the first scaling factors $s_1$ and $s_2$ chosen for $w_1$ and $w_2$ are either equal or differ by $1$. In either case, they share at leat $t-1$ scaling factors, and we can use those for estimation of the distance. This gives us $k$ comparable samples, as desired.

More generally, we can pick an integer $\tau \in [1,t]$, and set\footnote{In principle, $\beta$ can be chosen arbitrarily in $[\alpha^{1/\tau},\alpha^{1/\tau+1})$, and a value more conducive to floating point operations may be picked.} $\beta = \alpha^{1/\tau}$, We then use $\frac{k}{t-\tau}$ samples for each scales, and pick our scales starting from $s=\lceil \log_{1/\beta} \frac{Lk}{(t-\tau)w(U)}\rceil$.  It is easy to verify that this choice ensures that for any pair of sets with  $w_1(U)/w_2(U) \in [\alpha,1/\alpha]$, we are guaranteed to find $t$ comparable samples.

We formalize the sketching and the estimating algorithms next. We assume access to a subroutines {\em UnwtdSketch($S,k,r$)} that takes in an unweighted set $S$, parameter $k$ and a seed $r$ and returns a sketch consisting of $k$ samples. In addition, we assume that {\em UnwtdEstimate($sketch,sketch'$)} outputs an estimate of the Jaccard distance based on the sketches. The $b$-bit hashing scheme of~\cite{LiK11} would give a candidate pair of instantiations of these subroutines. In the description below, we assume that the randomness source $r$ can be partitioned into sources $r_i$, $r'_i$, for $t$ different values of $i$.

\begin{algorithm}[htb!]
\caption{ComputeSketch($w,k,\alpha,\tau,L,t,r$)}
\begin{algorithmic}[1]
\State $\beta = \alpha^{\frac{1}{\tau}}$
\State s = $\ceil{\log_{1/\beta} \frac{Lk}{(t-\tau)|w|_1}}$
\State sketch = $\emptyset$
\Statex 

\For {$i = s$ to $s+t-1$}       
	\State $w_i = \beta^{-i} w$   \Comment{Scale}
	\State $S_i = ReduceToUnwtd(w_i,r_i)$ \Comment{Round}
	\State $sk_i = UnwtdSketch(S_i, \frac{k}{t-\tau},r'_i)$ \Comment{Sketch}
	\State Add $(i,sk_i)$ to sketch
\EndFor
	
\State Output sketch
\end{algorithmic}
\label{alg:hmtsketch}
\end{algorithm}

\begin{algorithm}[htb!]
\caption{EstimateJaccard($sketch,sketch'$)}
\begin{algorithmic}[1]

\State Common = $\{(i,sk_i,sk'_i): (i,sk_i) \in sketch $ and $(i,sk'_i) \in sketch' \}$
\If {Common == $\emptyset$}
	\State Output "Similarity $< \alpha$"
\Else
	\For {each $(i,sk_i,sk'_i) \in$ Common}
		\State $sim_i = UnwtdSimilarity(sk_i,sk'_i)$
	\EndFor
	\State Output Average($sim_i$)
\EndIf

\end{algorithmic}
\label{alg:hmtestimate}
\end{algorithm}


We note that as stated, the number of hash computations required in the reduction steps is $tn$. However, using the dependent version $ReduceToUnwtdDep$ of the reduction, this can be reduced to $n$, which may be a substantial saving when $n$ is large.

While our algorithm can be used with any unweighted similarity estimation scheme, using it with the one permutation hashing scheme of~\cite{LiOZ12} will give us unweighted sets of expected size at most $\beta^{-1} \frac{Lk}{t-\tau},\ldots, \beta^{-t} \frac{Lk}{t-\tau}$. The number of hash evaluations in the unweighted sketching scheme is equal to the set size, so that for the typical setting of $\beta=\alpha=0.5$, $t=3$, the total cost is $7Lk$ hash evaluations. The running time is of a similar order. Here $L$ is a small constant such as $4$. Recall that the best previously known weighted scheme required $5nk$ hash evaluations, which we are reducing to $n+7Lk$.

\subsection{Discussion}
In this section, we discuss some finer implementation details, and how they may affect our choices of parameters.
\subsection*{The benefits of tunable unweighted size}
One immediate benefit of the set size being tunable is that we can arrange the parameters so as to ensure that at each of the scales that a set is involved in, the size of the unweighted set resulting from the reduction is at least $\frac{Lk}{t-\tau}$. As discussed earlier, this has the benefit of making empty samples rare enough that they can be ignored. This does not just simplify the sketch comparison but more importantly relieves us from having to store whether or not a sample is empty, leading to noticeable savings in the sketch size.

We remark that even for unweighted sets, the one permutation approach of~\cite{LiOZ12} suffers when sets are small. With many bins being empty the accuracy of the scheme drops drastically. Subsequent works~\cite{ShrivastavaL14a,ShrivastavaL14b} have proposed modifications to address this issue. Nevertheless, the authors still pay for sparseness: the variance initially falls off as $\frac{1}{k}$ as $k$ increases, but flattens out once $k$ becomes much larger than the set size. We note that treating an unweighted set as a weighted one, and applying our approach gives a simple solution to this problem, and results in a $1/k$ fall in variance for arbitrarily large $k$, irrespective of the support size, without paying any penalty in the running time. Thus, even for unweighted sets, the approach proposed in this work is useful.

A more subtle benefit comes from the fact that the size of the unweighted set is {\em at most} $\beta^{-t} \frac{Lk}{t-\tau}$, so that an average sample gets at most $L\beta^{-t}$ items. For typical values, $L=5, t=3, \beta = 0.5$, this is $40$. Thus when computing the hash value to compute the minimum in the bin, we can do with, say a 13 bit hash value. Using these many bits makes it exceedingly unlikely that one of the bins will not have a unambiguous minimum. Thus when generating hash values for $(a,1)\ldots,(a,j)$ for some $a$, we can use $\log_2 \frac{k}{t-\tau}$ bits to generate the bin, 13 bits to figure out the minimum, and an additional $b$ bits to be stored for the winner in the bin. Thus we need $7+13+2=22$ bits per $(a,i)$. Thus we can generate a sequence of pseudorandom bits seeded with $a$, and then break it up into chunks of 22 bits each, using the $i$th chunk for $(a,i)$. In the rare event that we do not get a unique minimum in some bin, we can reseed with $(a,1)$ and generate several additional bits per $i$ sequentially, and so on. Since this is a rare enough event it does not impose a significant cost. Note that in contrast, without such an upper bound on the required number of bits, the bucketing of any large document would lead to a large number of elements landing in the same bin which would require a larger number of bits to identify the minimum. One thus has to either reseed for each $(a,i)$ or make other assumptions on the document size.

\subsection*{The effect of the threshold for interestingness}

The threshold $\alpha$ which determines what values of Jaccard similarity we consider interesting would typically depend on applications. While we presented this work with $\alpha=0.5$, lower or higher similarity values may be preferable in other settings.
When $\alpha$ is close to 1 (say 0.95), then other optimizations may be possible. Indeed, note that for sets that are so similar, the sketches, even for a large values of $b$ would agree in nearly all the locations. Intuitively, each stored values gives us little information as there is at most say $5$ out of $100$ $b$-bit values that are different. One could ameliorate this by compressing the sketch in a careful manner. For example, one can take 10 $b$-bit values and just store their XOR, thus saving a factor of 10. In return, we can now generate a 1000 $b$-bit values instead of 100, but store only the 100 resulting XORs. For similarity more than 0.95, at least half of the 10-bin-blocks would be identical, and thus their XOR would be the same. Since we can account for accidental collision of the $b$-bit values, we can account for them. A careful look at this process shows that we would get an estimate of $Sim^{10} \approx (1-10(1-Sim))$ for $Sim$ close to 1, from which a more accurate estimate of the similarity can be obtained. This idea is not new and has been suggested in Li and K\"onig~\cite{LiK11}, who show that taking $b$ to be $1$ is already better when similarity is at least $0.5$, and show that xoring pairs (the $b= \frac{1}{2}$) case) gives a further improvement for larger similarities. It is natural to pick the appropriate value of $b<1$ when $\alpha$ is close to $1$.

\subsection*{Using more scales than 3}

Recall that with the proposed choice of 3 scales, we get two common scales for sets whose weights are within a factor of $\alpha$ of each other. But even for sets whose weights are within a factor of $\alpha^2$, we have one common scale and thus get a similarity estimate with error commensurate with $k/2$ samples instead of $k$.
By choosing more scales, we get a better trade-off between space and accuracy for every similarity value: we store at many scales, but have fewer samples for each scale. As $t$ increases, a larger fraction of our samples $(1-\frac{1}{t})$ are shared between two sets within weight $\alpha$. And we get a smoother fall-off in accuracy for smaller values of similarity: e.g., even sets with weights within a factor of $\alpha^2$ have $t-2$ scales in common, and thus a $(1-\frac{2}{t})$ of our samples can be used for estimating similarity. By setting $\tau$ to be larger, say $\tau = 3$, and larger $t$ (say 10), we get as much accuracy as the $\tau=1,t=3$ case for similarities around 0.5, but a smoother decay in accuracy for smaller similarities, and in fact a higher accuracy for higher similarities. The only way that we may have to ``pay'' for this is that as the size of the unweighted sets now becomes smaller the error in the reduction step may increase. In our experiments, even for unweighted sets of size 50, the error introduced was only a few percent, and moreover the errors over different scales seemed to largely cancel each other out.

Finally, we note that while the variance of the similarity estimate from different scales varies slightly. For the larger scales, the scaled weight is larger, so that the reduction has smaller variance. While at most a factor of $\frac{1}{L}$ of the variance from the unweighted sketch, this small difference in variance of the estimators from each scales can be taken into account while averaging the estimates from different scales. This would give a small improvement in the variance of the final estimate, at the cost of a slightly more complex estimation algorithm.

\section{Proofs}

In this section, we prove Theorems~\ref{thm:reduction} and~\ref{thm:reductiondep}. \subsection{A Useful Lemma}
We start by stating and proving a useful inequality that relates the expectation of the inverse of a sum of independent $0$-$1$ variables to the inverse of a closely related expectation. This is a special case of a result of \cite{ChaoS72} and we present a simple proof here for completeness.

\begin{lemma}
\label{lem:expinverse}
Let $\{X_i\}_{i=1}^N$ be a sequence of independent Bernoulli random variables with $\E[X_i]=\mu_i$ and let $\mu \stackrel{def}{=} \sum_i \mu_i$. Then for $A\geq 1$,
$$
\frac{1}{A+\mu} \leq \E\left[(A+\sum_{i=1}^N X_i)^{-1}\right]  \leq \frac{1}{A+\mu-1}.
$$
\end{lemma}
\begin{proof}
The first inequality follows by applying Jensen's inequality to the function $\phi(X)=1/(A+X)$, which is convex for $X\geq 0$.

To prove the second inequality, we use a result from~\cite{ChaoS72} who gave a formula for negative moments of random variables. We reproduce the proof of the case we use for completeness. Observe that for every $t,x>0$,
$$
\frac{t^x}{x} = \int_0^t u^{x-1} \ud u.
$$
Setting $t=1$ and taking expectations over random $x$, we get
$$
\E\big[\frac{1}{X}\big] = \int_0^1 \E\big[u^{X-1}\big] \ud u.
$$
For $X= A+ \sum_i X_i$, we upper bound
\begin{align*}
\E\big[u^{X+A-1}\big]
&= u^{A-1}\prod_i \E\big[u^{X_i}\big]\\
&= \exp((A-1)\ln u)\prod_i (1-(1-u)\mu_i)\\
&\leq \exp((A-1)(u-1))\prod_i \exp(-(1-u)\mu_i)\\
&= \exp(-(1-u)(\mu+A-1)).
\end{align*}
We conclude that
\begin{align*}
\E\big[\frac{1}{X}\big] &\leq \exp(-(\mu+A-1))\int_{0}^1 \exp((\mu+A-1) u)\ud u \\
&= \frac{1-\exp(-(\mu+A-1))}{\mu+A-1}<\frac{1}{\mu+A-1}.
\end{align*}
\end{proof}

\subsection{Proof of Theorem~\ref{thm:reduction}}

We set up some notation first. Given weighted sets $w_1$ and $w_2$ , we define $w_{\min} : U \to \Re$ as $$w_{\min}(a)= \min(w_1(a),w_2(a))$$ and similarly define $w_{\max} : U \to \Re$ as $$w_{\max}(a) = \max(w_1(a),w_2(a)).$$  Recall that $S_i = ReduceToUnwtd(w_i,r)$, and let $S_{\min}$ and $S_{\max}$ denote the outcomes $ReduceToUnwtd(w_{\min},r)$ and $ReduceToUnwtd(w_{\max},r)$ respectively.

Our threshold-based rounding has the property that there is no loss of generality in assuming that $w_1=w_{\min}$ and $w_2=w_{\max}$. This is because the Jaccard similarity $jacc(w_1,w_2)$ equals the similarity $jacc(w_{\min},w_{\max})$, and moreover for any value of $r$ we have that $(a,j) \in S_1 \cap S_2$ if and only if $(a,j) \in S_{\min}$, and similarly $(a,j) \in S_1 \cup S_2$ if and only if $(a,j) \in S_{\max}$. Therefore, $jacc(S_1,S_2) = jacc(S_{\min},S_{\max})$ for every value of $r$. For the rest of this proof we will therefore assume that $w_1=w_{\min}$ and $w_2=w_{\max}$.

Let $Z_1 : U \to \{0,1\}$ be the indicator function for the set $S_1$ and similarly define $Z_2$. Note that both $Z_1$ and $Z_2$ are random variables. For a function $f: U \to \Re$, let $f(U)$ denote $\sum_{a \in U} f(a)$.

Part (1) of Theorem~\ref{thm:reduction} is now immediate by linearity of expectation, since for $i=1,2$ we have that $$\Ex[S_i] = \sum_{a \in U} \sum_{j\in \N} \Ex[Z_i(a,j)] = \sum_{a \in U} w_1(a).$$

Part (2) of Theorem~\ref{thm:reduction} follows by a direct application of Chernoff bounds (see e.g.~\cite{DubhashiP09}). 

We now proceed to proving parts (3) and (4) of Theorem~\ref{thm:reduction}. Let us define $w_{\rest} = w_{2} - w_{1}$ and similarly $Z_{\rest} = Z_{2} - Z_{1}$. Let $\hat{U} = U \times \N$ and we will denote a generic element of $\hat{U}$ by $e = (a,j)$. With this notation the original weighted Jaccard similarity equals $$jacc(w_1,w_2) = w_{1}(U)/w_{2}(U).$$ The Jaccard similarity between $S_1$ and $S_2$ on the other hand is $$jacc(S_1,S_2) = jacc(Z_1,Z_2) = Z_{1}(\hat{U})/Z_{2}(\hat{U}).$$   Also, note that for $i=1,2$ and for each $a\in U$, the random variables $Z_i(a,j)$ are all Bernoulli random variables.

For $e=(a,j)$, define $X_e = Z_1(e)$, and $Y_e= Z_2(e)-Z_1(e)$. These random variables then satisfy the following properties:
\begin{itemize}
\item $X_e$ and $Y_e$ are Bernoulli random variables.
\item The random variables $\{(X_e,Y_e): e \in \hat{U}\}$ are independent of each other. Thus, $X_e$ may depend on $Y_e$ but not on $X_{e'}$ for $e \neq e'$.
\item For any $e$, at least one of $X_e$ and $Y_e$ is zero. Thus, $\Ex[X_eY_e]=0$.
\item $\sum_{e \in \hat{U}} \Ex[X_e] = w_1(U)$. 
\item $\sum_{e \in \hat{U}} \Ex[X_e+Y_e] = w_2(U)$.
\end{itemize}

We first prove part (4). This is an easy consequence of Chernoff bounds applied to the sums of Bernoulli random variables $X_e$, and to the sum of $(X_e+Y_e)$.  Let $X$ denote $\sum_e X_e$ and $Y$ denote $\sum_e Y_e$. Let $\mu_x =E[X]=w_1(U)$ and $\mu_y=E[Y]  = w_2(U)-w_1(U)$. 

Without loss of generality we can assume that $E[X] \geq E[Y]$ (or else can argue about $\frac{Y}{X+Y}$). 
Now by standard Chernoff bounds,
$$\Pr[|X-\mu_x|\geq \alpha \mu_x] \leq 2\exp(-\alpha^2 \mu_x/3)$$
and
\begin{align*}
\Pr[|X+Y-\mu_x-\mu_y| \geq \alpha &(\mu_x+\mu_y)] \\
&\leq  2\exp(-\alpha^2 (\mu_x+\mu_y)/3).
\end{align*}
Thus, except with probability $4\exp(-\alpha^2\mu_x/3)$, we have 
$$X/\mu_x \in (1-\alpha,1+\alpha)$$
and 
$$(X+Y)/(\mu_x+\mu_y) \in (1-\alpha,1+\alpha).$$
In this case, 
$$\frac{X}{X+Y}/\frac{\mu_x}{\mu_x+\mu_y} \in (\frac{1-\alpha}{1+\alpha},\frac{1+\alpha}{1-\alpha}) \subset (1-2\alpha,1+3\alpha).$$
Thus, $|\frac{X}{X+Y} - \frac{\mu_x}{\mu_x+\mu_y}| \leq 3\alpha$ and setting $\alpha = \sqrt{\frac{3\ln 4/\delta}{\mu_x}}$ implies that

\begin{align*}
\Pr[|\frac{X}{X+Y} - \frac{\mu_x}{\mu_x+\mu_y}| \geq 3\sqrt{\frac{3\ln 4/\delta}{\mu_x}}] \leq \delta,
\end{align*}
 thus proving part (4). We remark that we did not attempt to optimize the constants here.

Finally, we will prove the following result, which implies part (3).
\begin{theorem} \label{thm:biastech}
Let $(X_1,Y_1),\ldots, (X_m,Y_m)$ be a sequence of independent tuples of Bernoulli random variables such that $\Ex[X_i] = p_i$, $\Ex[Y_i]=q_i$ and $\Ex[X_i Y_i]=0$ (i.e., they are never 1 together). Let $\mu_x = \sum_i p_i$ and $\mu_y = \sum_i q_i$.  Let $X=\sum_i X_i$ and $Y=\sum_i Y_i$. Then assuming\footnote{We use the convention that $0/0=1$ for the left inequality, and $0/0=0$ for the right one.} that $\mu_x+\mu_y >1$,
\begin{align*}
\frac{\mu_x-1}{\mu_x+\mu_y - 1} \leq \Ex[\frac{X}{X+Y}] \leq \frac{\mu_x}{\mu_x+\mu_y - 1}.
\end{align*}
\end{theorem}
\begin{proof}


We first write
\begin{align*}
\Ex&[\frac{X}{X+Y}] \\
&=  \Ex_{(X_1,\ldots,X_n)}[X\cdot \Ex_{Y_1,\ldots,Y_n}[\frac{1}{X +Y}\mid X_1,\ldots,X_n],
\end{align*}

When $X=0$, the expression inside the outer expectation is zero. For $X \geq 1$, we apply Lemma~\ref{lem:expinverse} to conclude that
\begin{align*}
X\Ex_{(Y_1,\ldots,Y_n)}&[\frac{1}{X+Y}| X_1,\ldots,X_n]\\
 &\leq \frac{X}{X+\Ex[Y|X_1,\ldots,X_n]-1}.
\end{align*}
Now recall that $\Ex[Y_i|X_i=1]=0$. It follows that $\Ex[Y_i|X_i] = (1-X_i)q_i/(1-p_i)$.
Thus, 
\begin{align*}
\Ex&[\frac{X}{X+Y}] = \\
&=\Ex_{(X_1,\ldots,X_n)}[\sum_i X_i/(\sum_i (X_i + \frac{(1-X_i)q_i}{1-p_i}) -1)]\\
&=\Ex_{(X_1,\ldots,X_n)}[\sum_i X_i/(\sum_i \gamma_i X_i + \sum_i \beta_i -1)
\end{align*}
where $\beta_i = \frac{q_i}{1-p_i}$ and $\gamma_i =1-\beta_i$. This expression is easily seen to be concave in each $X_i$ (for any fixing of the other $X_j$'s). Indeed denoting the numerator by $f$ and the denominator by $g$, the partial derivative $\pder{f/g}{X_i} = \frac{1}{g} - \frac{\gamma_i f}{g^2}$ and $\pdertwo{f/g}{X_i} = -\frac{2\gamma_i}{g^2}(1-\gamma_i\cdot \frac{f}{g})$. Since both $f/g$ and $\gamma_i$ are at most 1, the concavity follows. Thus, using Jensen's inequality, we can one-by-one replace the random variable $X_i$ by its expectation. Rearranging we get
\begin{align*}
\Ex[\frac{X}{X+Y}] 
&\leq \frac{\sum_i \Ex[X_i]}{\sum_i \gamma_i \Ex[X_i] + \sum_i \beta_i -1}
&= \frac{\mu_x}{\mu_x+\mu_y-1},
\end{align*}
which implies the second inequality. By symmetry
\begin{align*}
\Ex[\frac{Y}{X+Y}] &\leq  \frac{\mu_y}{\mu_x+\mu_y - 1}
\end{align*}
Noting that $\frac{X}{X+Y} = 1- \frac{Y}{X+Y}$ then implies the first inequality.
\end{proof}


This implies that both $\Ex[jacc(S_1,S_2)]$ and $jacc(w_1,w_2) = \frac{w_1}{w_2}$ are sandwiched in the interval $[\frac{w_1-1}{w_2-1},\frac{w_1}{w_2-1}]$. Since this interval is of size $\frac{1}{w_2-1}$, this implies part (3) and completes the proof of Theorem~\ref{thm:reduction}.

\subsection{Proof of Theorem~\ref{thm:reductiondep}}

In this section, we prove Theorem~\ref{thm:reductiondep}.
We show that even if the random threshold for rounding $(a,j)$ is chosen as $h(r,a)$ instead of $h(r,a,j)$, the properties of the reduction hold. Once again, there is no loss of generality in assuming that $w_1=w_{\min}$ and $w_2=w_{\max}$. We note that for $i=1,2$,  and for any $a$, the variables $Z_i(a,j)$ are all deterministic except for $Z_i(a,\lceil w_{i}(a)\rceil)$. Thus, for each $i$, the Bernoulli random variables $\{Z_i(e): e \in U \times \N\}$ are all independent. This is sufficient for the proofs of parts (1), (2) and (4) to go through unchanged. It remains to argue that part (3) still holds.


To prove this, we will need to handle additional dependencies between the tuples $(X_e,Y_e)$ as defined above. For a fixed $a$, let $e_a$ denote $(a,\lceil w_1(a)\rceil)$, and $e'_a$ denote $(a,\lceil w_2(a)\rceil)$. If $e_a=e'_a$, then the only non-deterministic random variable amongst $\{(X_e,Y_e): e=(a,j)\}$ is the tuple $(X_{e_a},Y_{e_a})$. If on the other hand, $e_a < e'_a$, then $X_{e'_a}$ is deterministically $0$ and $Y_{e_a}=1-X_{e_a}$.  Moreover the random variables $X_{e_a}$ and $Y_{e'_a}$ are correlated through a common threshold $thresh$, with $X_{e_a} = \one(thresh<w_2(a) - \lfloor w_2(a)\rfloor)$ and $Y_{e'_a}= \one(thresh<w_2(a) - \lfloor w_2(a)\rfloor)$. 

For any $e$, let $Y'_e$ denote $Y_{e'_a}$ if $e=e_a$ for some $a$, and let $Y_e$ be deterministically $0$ otherwise. For $e=e'_a\neq e_a$ we redefine $Y_e$ to be $0$. This essentially moves the  troubling random variable $Y_{e'_a}$ to $Y'_{e_a}$. Furthermore, this regrouping ensures that all these triples $\{(X_e, Y_e,Y'_e): e \in U \times \N\}$ are independent of each other. The next result is an analog of Theorem~\ref{thm:biastech}, albeit with a more complex proof.


\begin{theorem} \label{thm:biastechdep}
Suppose $(X_1,Y_1,Y'_1),\ldots, (X_n,Y_n,Y'_n)$ is a se\-quence of independent tuples of Bernoulli random variables such that $\Ex[X_i] = p_i$, $\Ex[Y_i]=q_i$, $\Ex[Y'_i]=q'_i$, and $\Ex[X_i Y_i]=0$, i.e., $X_i$ and $Y_i$ are never 1 simultaneously. Further suppose that either (a) $q'_i=0$ or (b) $q_i=1-p_i$, $X_{i} = \one(thresh<p_i)$ and $Y'_{i}= \one(thresh<q'_i)$ for threshold $thresh$ chosen uniformly in $[0,1]$. 
 Let $X=\sum_i X_i$, $Y=\sum_i (Y_i+Y'_i)$, $\mu_x = \Ex[X]$ and $\mu_y=\Ex[Y]$. Then assuming\footnote{We use the convention that $0/0=1$ for the left inequality, and $0/0=0$ for the right one.} that $\mu_x+\mu_y >1$,
\begin{align*}
\frac{\mu_x-1}{\mu_x+\mu_y - 1} \leq \Ex[\frac{X}{X+Y}] \leq \frac{\mu_x}{\mu_x+\mu_y - 1}.
\end{align*}
\end{theorem}
\begin{proof}


We first write
\begin{align*}
\Ex&[\frac{X}{X+Y}] \\
&=  \Ex_{(X_1,\ldots,X_n)}[X\cdot \Ex_{Y_1,Y'_1\ldots,Y_n,Y'_n}[\frac{1}{X +Y}\mid X_1\ldots,X_n],
\end{align*}

When $X=0$, the expression inside the outer expectation is zero. For $X \geq 1$, we once again apply Lemma~\ref{lem:expinverse}, which applies since conditioned on $X_i$, either $Y_i$ is fully determined, or $Y'_i$ is deterministically zero. In either case, at most one of $Y_i,Y'_i$ is random. We conclude that
\begin{align*}
X\Ex_{(Y_1,Y'_1\ldots,Y_n,Y'_n)}&[\frac{1}{X+Y}| X_1,\ldots,X_n]\\
 &\leq \frac{X}{X+\Ex[Y|X_1,\ldots,X_n]-1}.
\end{align*}

We now need to estimate these expectations. We will compute $\Ex[X_i+Y_i+Y'_i|X_i]$. We consider several cases. In each case, we will show that this expectation can be written as $\beta_i + \gamma_i X_i$, with $\gamma_i \leq 1$. 

Case 1: $q'_i=0$. This case is similar to the setting of Theorem~\ref{thm:biastech}. Here $\Ex[Y_i|X_i=1]=0$. It follows that $\Ex[X_i+Y_i+Y'_i|X_i] = X_i+(1-X_i)q_i/(1-p_i) = \frac{q_i}{1-p_i} + X_i\frac{1-p_i-q_i}{1-p_i}$.  Clearly, $\gamma_i \leq 1$.

Case 2(a): $q_i=1-p_i, q'_i \geq p_i$. In this case, $X_i=\one(thresh<p_i)$, so that $X_i=$ implies that $Y'_i=1$ as well. If $X_i=0$, then $Y'_i$ is $1$ with probability exactly $(q'_i-p_i)/(1-p_i)$. Moreover, $Y_i = 1-X_i$. Thus in this case, $\Ex[X_i+Y_i+Y'_i|X_i] = 1 + X_i + (1-X_i)(q'_i-p_i)/(1-p_i) = 1 + \frac{q'_i-p_i}{1-p_i} + X_i\frac{1-q'_i}{1-p_i}$. By assumption $\gamma_i \leq 1$.

Case 2(b): $q_i=1-p_i, q'_i < p_i$. In this case $X_i=0$ iff $thresh>p_i$ in which case $Y'_i=0$ as well. If $X_i=1$, then $Y'_i$ is $1$ with probability exactly $q'_i/p_i$. Once again, $Y_i=1-X_i$. Thus in this case, $\Ex[X_i+Y_i+Y'_i|X_i] = 1 + X_i\frac{q'_i}{p_i}$. By assumption $\gamma_i < 1$.

Thus in all cases, we can write 
\begin{align*}
\Ex&[\frac{X}{X+Y}] = \\
&=\Ex_{(X_1,\ldots,X_n)}[\sum_i X_i/(\sum_i \gamma_i X_i + \sum_i \beta_i -1)
\end{align*}
with $\gamma_i \leq 1$. This expression is then easily seen to be concave in each $X_i$. Indeed denoting the numerator by $f$ and the denominator by $g$, the partial derivative $\pder{f/g}{X_i} = \frac{1}{g} - \frac{\gamma_i f}{g^2}$ and $\pdertwo{f/g}{X_i} = -\frac{2\gamma_i}{g^2}(1-\gamma_i\cdot \frac{f}{g})$. Since both $f/g$ and $\gamma_i$ are at most 1, the concavity follows. Thus using Jensen's inequality,
\begin{align*}
\Ex[\frac{X}{X+Y}] 
&\leq \frac{\sum_i \Ex[X_i]}{\sum_i \gamma_i \Ex[X_i] + \sum_i \beta_i -1}.
\end{align*}

It remains to compute the value of the denominator. Since $\Ex[X_i+Y_i+Y'_i|X_i] = \beta_i+\gamma_i X_i$, it follows that $\Ex[X_i+Y_i+Y'_i] = \beta_i + \gamma_i \Ex[X_i]$. Thus the denominator is exactly $\Ex[X+Y] - 1$.  This implies the second inequality.

To prove the first inequality, it will once again suffice to argue that $\E[\frac{Y}{X+Y}]$ is bounded above by $\frac{\mu_y}{\mu_x +\mu_y -1}$. Unlike Theorem~\ref{thm:biastech}, the variables $X$ and $Y$ are not symmetric. The proof however is relatively straightforward and we sketch it next.  We will now write 
\begin{align*}
\Ex&[\frac{Y}{X+Y}] \\
&=  \Ex_{(Y_1,Y'_1\ldots,Y_n,Y'_n)}[Y\cdot \Ex_{X_1,\ldots,X_n}[\frac{1}{X +Y}\mid Y_1,Y'_1\ldots,Y_n,Y'_n],
\end{align*}

When $Y=0$, the expression inside the outer expectation is zero. For $Y \geq 1$, we once again apply Lemma~\ref{lem:expinverse} to write
\begin{align*}
Y\cdot &\Ex_{X_1,\ldots,X_n}[\frac{1}{X +Y}\mid Y_1,Y'_1\ldots,Y_n,Y'_n] \\
 &\leq \frac{Y}{Y+\Ex[X|Y_1,Y'_1\ldots,Y_n,Y'_n]-1}.
\end{align*}
We once again handle the two cases separately. In the case that $q_i=1-p_i$, the variable $X_i$ is fixed given $Y_i$, and the term $Y_i+Y'_i + \Ex[X_i|Y_i,Y'_i]$ is equal to $1+Y'_i$ (i.e., a deterministic quantity under the conditioning). If on the other hand, $q'_i=0$, then $Y'_i$ is deterministically 0, and the term $Y_i+Y'_i + \Ex[X_i|Y_i,Y'_i]$  can be written as $\beta_i + \gamma_i Y_i$ with $\gamma_i \leq 1$. Thus we can apply Jensen's inequality to derive the claimed bound.
\end{proof}

This then implies part (3) of Theorem~\ref{thm:reductiondep} and completes its proof.

\section{Synthetic Test Results}

\makeatletter
\newenvironment{tablehere}
{\def\@captype{table}}
{}
\newenvironment{figurehere}
{\def\@captype{figure}}
{}
\makeatother

We produced pseudo-random artificial sets of weights, which we modified to center around a few chosen levels of exact weighted Jaccard similarity: 95\%, 90\%, 85\%, 80\%, 70\%, 65\%, 60\%, 55\%,  50\%, and 40\%. We implemented a single scaling version of our algorithm, so that we could estimate Jaccard for all pairs, as well as implementing Ioffe's algorithm. We selected full-length samples, as well as 2-bit, 1-bit, and half-bit compressions; the number of samples was chosen from 64, 128, 256, and 512.


The first figures we present compare Ioffe sketching against our algorithm using exact Jaccard and our algorithm using binning, measuring the error between the values these algorithms produce, versus the underlying truth. In Figure~\ref{fig:IoffevsHMTaverageError}, the yellow bars show the average absolute error when estimating weighted Jaccard using 128 Ioffe samples, shown at a variety of underlying true Jaccard values ranging from 0.40 (on the right) to 0.96, while the blue bars show the estimate using our algorithm computing an estimated 1024 samples randomly assigned to 128 bins. The green bars show the average absolute error introduced by randomized rounding, in which every input item is reassigned to an integer close to its true weight, and the Jaccard value of these roundings are computed exactly.  We display these because, following the theorems above on bias, all bias away from true Jaccard is introduced by rounding; both Ioffe's sampling, and binning preserve expected Jaccard values. At high true Jaccard values both Ioffe and we perform well, missing a true Jaccard value of 0.96 by approximately 0.01; this corresponds to getting a mismatch in slightly over one bin; due simply to quantization in 128 bins, we have to expect an error of at least one in 256. We observe that our algorithm is insignificantly worse than Ioffe at very high Jaccard values from 0.9 to 1.0 (largely due to errors introduced by rounding. Our algorithm is repeatably somewhat better than Ioffe at Jaccard values between 0.8 and 0.9, and is roughly equivalent for Jaccard values down to 0.5, below which point Ioffe sampling is better than our technique, although the absolute error for our technique never exceeds 0.035, corresponding to getting an average excess mismatch of approximately 4 bins.

The next figure, Figure~\ref{fig:IoffevsHMTaverageError} shows the same bars, but with a new color assignment (Ioffe is now blue, our algorithm is now red, and randomized rounding is gray), and computes the standard deviation of the observation made above corresponding to each true Jaccard value.  In this graph we see that our standard deviation is smaller than Ioffe's except at the lowest true Jaccard value, at which point the mismatch in scalings starts to dominate the computation.  We also see that the standard deviation is nearly as large as the absolute error, mostly coming from estimates which are closer to the true value.

\begin{figurehere}
 \centering
 \includegraphics[width=0.8 \columnwidth]{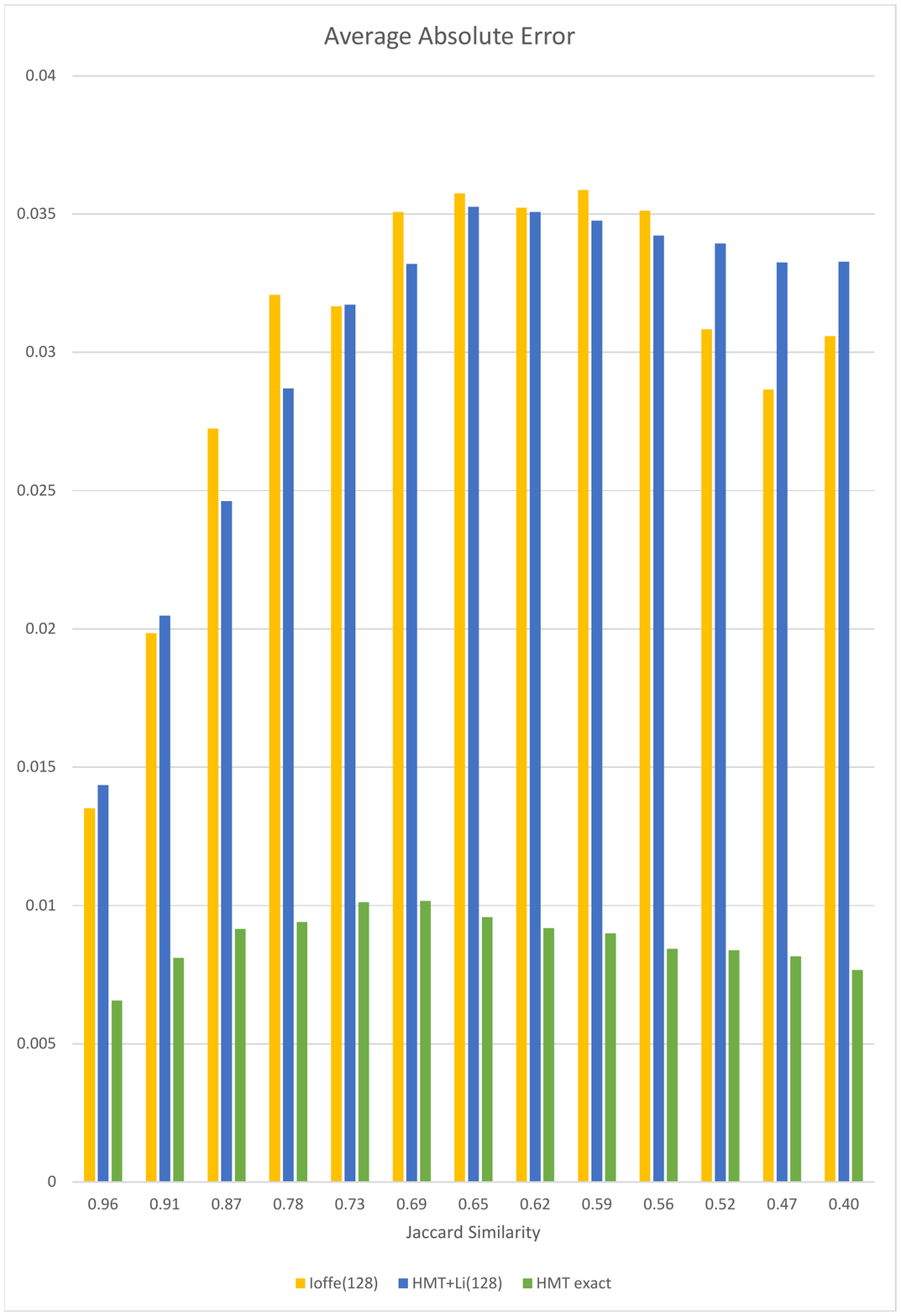}
 \caption{Comparison of Ioffe's and our sampling in terms of average error}
\label{fig:IoffevsHMTaverageError}
 \end{figurehere}

\begin{figurehere}
 \centering
 \includegraphics[width=0.8 \columnwidth]{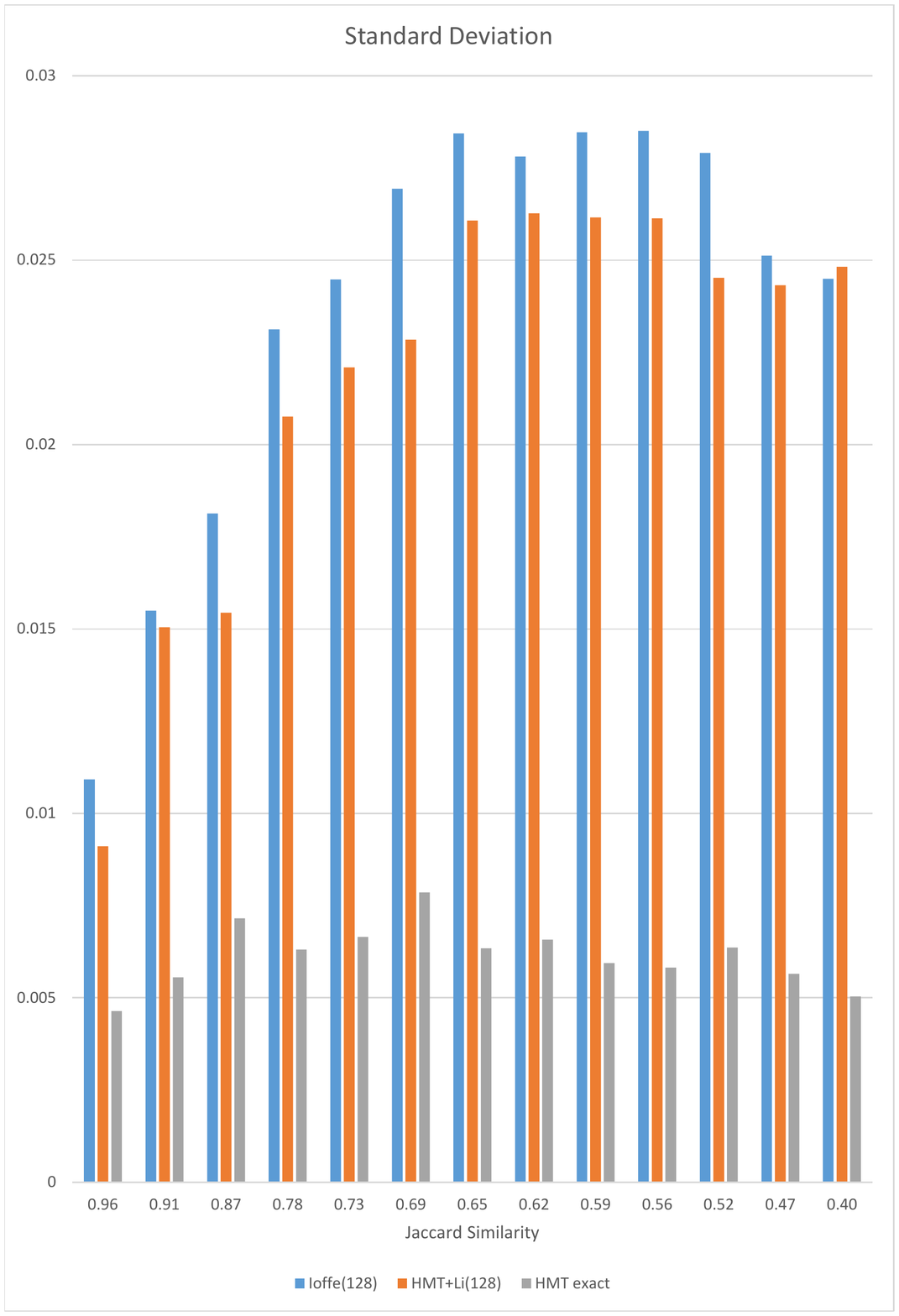}
 \caption{Comparison of Ioffe's and our sampling in terms of standard deviation}
\label{fig:IoffevsHMTdeviation}
 \end{figurehere}

\nocite{LiSK13,AlonsoFM13,TheobaldSP08}

\section{Conclusions}
We have presented a simple scheme to reduce weighted sets to unweighted ones efficiently, in such a way that the size of the resulting unweighted set is tunable, and the Jaccard between sets of comparable sizes is preserved very accurately. We have shown how to use this scheme for the problem of building sketches for Jaccard similarity estimation for weighted sets. The resulting scheme is two orders of magnitude faster than previously known schemes for typical setting of parameters, and does not suffer any significant loss in quality. We prove that the scheme has a non-zero but negligible bias, and satisfies tail inequalities similar to the unweighted case. We also show empirical results showing that this computational benefit comes at negligible cost in accuracy in the interesting case of large similarity.
\bibliographystyle{abbrv}
\bibliography{jaccard}

\begin{thebibliography}{10}

\bibitem{AilonC09}
N.~Ailon and B.~Chazelle.
\newblock The fast {J}ohnson-{L}indenstrauss transform and approximate nearest
  neighbors.
\newblock {\em SIAM Journal of Computing}, 39(1):302--322, 2009.

\bibitem{AilonL09}
N.~Ailon and E.~Liberty.
\newblock Fast dimensionality reduction using {R}ademacher series on dual {BCH}
  codes.
\newblock {\em Discrete and Computational Geometry}, 42(4):615--630, 2009.

\bibitem{AlonsoFM13}
O.~Alonso, D.~Fetterly, and M.~Manasse.
\newblock Duplicate news story detection revisited.
\newblock In {\em Information Retrieval Technology}, volume 8281 of {\em
  Lecture Notes in Computer Science}, pages 203--214. 2013.

\bibitem{Broder97}
A.~Z. Broder.
\newblock {On the resemblance and containment of documents}.
\newblock In {\em Compression and Complexity of Sequences}, pages 21--29, 1997.

\bibitem{BroderCFM98}
A.~Z. Broder, M.~Charikar, A.~M. Frieze, and M.~Mitzenmacher.
\newblock Min-wise independent permutations.
\newblock In {\em STOC}, pages 327--336, 1998.

\bibitem{BroderGMZ97}
A.~Z. Broder, S.~C. Glassman, M.~S. Manasse, and G.~Zweig.
\newblock {Syntactic Clustering of the Web}.
\newblock {\em Computer Networks and Isdn Systems}, 29:1157--1166, 1997.

\bibitem{ChaoS72}
M.~T. Chao and W.~E. Strawderman.
\newblock Negative moments of positive random variables.
\newblock {\em Journal of the American Statistical Association},
  67(338):429--431, 1972.

\bibitem{Charikar02}
M.~S. Charikar.
\newblock Similarity estimation techniques from rounding algorithms.
\newblock In {\em Proceedings of the ACM Symposium on Theory of Computing
  (STOC)}, pages 380--388, 2002.

\bibitem{DasguptaKS10}
A.~Dasgupta, R.~Kumar, and T.~Sarl\'os.
\newblock A sparse {J}ohnson-{L}indenstrauss transform.
\newblock In {\em Proceedings of the ACM Symposium on Theory of Computing
  (STOC)}, pages 341--350, 2010.

\bibitem{DubhashiP09}
D.~P. Dubhashi and A.~Panconesi.
\newblock {\em Concentration of Measure for the Analysis of Randomized
  Algorithms}.
\newblock Cambridge University Press, 2009.

\bibitem{FlajoletM85}
P.~Flajolet and G.~Martin.
\newblock {Probabilistic counting for database applications}.
\newblock {\em Journal of Computer and System Sciences}, 1985.

\bibitem{GollapudiP}
S.~Gollapudi and R.~Panigrahy.
\newblock Exploiting asymmetry in hierarchical topic extraction.
\newblock In {\em Proceedings of the ACM Conference on Information and
  Knowledge Management (CIKM)}, pages 475--482, 2006.

\bibitem{Henzinger06}
M.~R. Henzinger.
\newblock Finding near-duplicate web pages: a large-scale evaluation of
  algorithms.
\newblock In {\em Annual ACM SIGIR Conference}, pages 284--291, 2006.

\bibitem{Ioffe10}
S.~Ioffe.
\newblock Improved consistent sampling, weighted minhash and l1 sketching.
\newblock In {\em IEEE International Conference on Data Mining (ICDM)}, pages
  246--255, 2010.

\bibitem{LiK11}
P.~Li and A.~C. K\"onig.
\newblock {Theory and applications of $b$-bit minwise hashing}.
\newblock {\em Communications of The ACM}, pages 101--109, 2011.

\bibitem{li2010b}
P.~Li and C.~K{\"o}nig.
\newblock b-bit minwise hashing.
\newblock In {\em Proceedings of the ACM International World Wide Web
  Conference (WWW)}, pages 671--680, 2010.

\bibitem{LiOZ12}
P.~Li, A.~Owen, and C.-H. Zhang.
\newblock One permutation hashing.
\newblock In {\em Advances in Neural Information Processing Systems 25}, pages
  3122--3130. 2012.

\bibitem{LiSK13}
P.~Li, A.~Shrivastava, and A.~C. K\"onig.
\newblock $b$-bit min-wise hashing in practice.
\newblock In {\em Internetware}, 2013.

\bibitem{ManasseMT}
M.~Manasse, F.~McSherry, and K.~Talwar.
\newblock Consistent weighted sampling.
\newblock Technical Report MSR-TR-2010-73, 2010.

\bibitem{ShrivastavaL14a}
A.~Shrivastava and P.~Li.
\newblock Densifying one permutation hashing via rotation for fast near
  neighbor search.
\newblock In {\em Proceedings of the International Conference on Machine
  Learning (ICML)}, pages 557--565, 2014.

\bibitem{ShrivastavaL14b}
A.~Shrivastava and P.~Li.
\newblock Improved densification of one permutation hashing.
\newblock arXiv:1406.4784v1, 2014.

\bibitem{ShrivastavaL14c}
A.~Shrivastava and P.~Li.
\newblock In defence of minhash over simhash.
\newblock arXiv:1407.4416v1, 2014.

\bibitem{TheobaldSP08}
M.~Theobald, J.~Siddharth, and A.~Paepcke.
\newblock Spotsigs: robust and efficient near duplicate detection in large web
  collections.
\newblock In {\em Annual ACM SIGIR Conference}, pages 563--570, 2008.

\end{thebibliography}

\end{document}